\documentclass[a4paper]{article}

\usepackage{fullpage}
\usepackage{graphicx}

\usepackage{microtype}

\usepackage{amssymb}
\usepackage{amsmath}
\usepackage{amsthm}

\usepackage{xspace}
\usepackage{color}

\newtheorem{theorem}{Theorem}
\newtheorem{definition}[theorem]{Definition}
\newtheorem{lemma}[theorem]{Lemma}
\newtheorem{corollary}[theorem]{Corollary}

\let\leq\leqslant
\let\geq\geqslant
\let\le\leqslant
\let\ge\geqslant
\let\theta\vartheta

\usepackage{hyperref}
\hypersetup{%
  breaklinks,%
  colorlinks=true,%
  linkcolor=[rgb]{0.45,0.0,0.0},%
  citecolor=[rgb]{0,0,0.45}
}

\makeatletter
\partopsep\z@ \textfloatsep 10pt plus 1pt minus 4pt
\def\section{\@startsection {section}{1}{\z@}%
  {-3.5ex plus -1ex
    minus -.2ex}{2.3ex plus .2ex}{\large\bf}}
\def\subsection{\@startsection{subsection}{2}%
  {\z@}{-3.25ex plus
    -1ex minus -.2ex}{1.5ex plus .2ex}{\normalsize\bf}}
\def\@fnsymbol#1{\ensuremath{\ifcase#1\or 1\or 2\or 3\or 4\or
    5\or 6\or 7 \or 8\ or 9 \or 10\or 11 \else\@ctrerr\fi}}
\makeatother

\newcommand{\packAA}{\textsc{DiskPackingAABox}\xspace}
\newcommand{\packB}{\textsc{DiskPackingBox}\xspace}
\newcommand{\packC}{\textsc{DiskPackingConvex}\xspace}
\newcommand{\stab}{\textsc{DiskStabbing}\xspace}
\newcommand{\ds}{d_s} 
\newcommand{\OPT}{\textsc{opt}}

\newcommand{\hh}{\bar{h}}
\newcommand{\disks}{\mathcal{D}}

\newcommand{\eps}{\varepsilon}
\newcommand{\Reals}{\mathbb{R}}

\DeclareMathOperator{\vol}{vol}

\title{Packing $d$-dimensional balls into a $d+1$-dimensional container}

\author{Helmut Alt%
  \thanks{Institut f\"ur Informatik, Freie Universit\"at Berlin;
    {alt@mi.fu-berlin.de}, {nadja.seiferth@fu-berlin.de}}
  \and
  Sergio Cabello%
  \thanks{University of Ljubljana and Institute of Mathematics, 
		Physics and Mechanics; sergio.cabello@fmf.uni-lj.si}
  \and
  Otfried Cheong%
  \thanks{SCALGO; otfried@scalgo.com}
  \and
  Ji-won Park%
  \thanks{Universit\'e de Lorraine, CNRS, Inria, LORIA; Ji-won.Park@inria.fr}
  \and
  Nadja Seiferth\footnotemark[1]}

\begin{document}

\maketitle

\begin{abstract}
  In this article, we consider the problems of finding in~$d+1$
  dimensions a minimum-volume axis-parallel box, a minimum-volume 
  arbitrarily-oriented box and a minimum-volume convex body into which 
  a given set of $d$-dimensional unit-radius balls can be packed under
  translations.  
  The computational problem is neither known to be NP-hard nor to be
  in NP.  We give a constant-factor approximation algorithm for
  each of these containers based on a reduction to finding a 
  shortest Hamiltonian path in a weighted graph, which in turn models
  the problem of stabbing the centers of the input balls  
  while keeping them disjoint.
  We also show that for~$n$ such balls, a container of
  volume~$O(n^{\frac{d-1}{d}})$ is always sufficient and sometimes
  necessary.  As a byproduct, this implies that for~$d \geq 2$ there
  is no finite size $(d+1)$-dimensional convex body into which all 
  $d$-dimensional unit-radius balls can be packed simultaneously.
\end{abstract}

\section{Introduction}

Packing a set of geometric objects in a non-overlapping way into a
minimum-size container is an intriguing problem.  Because of its
practical significance, many variants of the problem have been
investigated, see the surveys~\cite{Bulletin,Scheithauer} and the
references therein. Even a simple variant of packing a set of
rectangles into a rectangular container turns out to be
NP-complete~\cite{FowlerPT81}.  In many cases, not much is known about
the true complexity of the problem.

Constant-factor approximation algorithms of polynomial running time
have been found for many variants of packing problems, in particular
for finding minimum-size rectangular or convex containers for a set of
convex polygons under translations~\cite{ConvexPolygons}, that is, the
objects may be translated but rotations are not allowed. Also,
approximation algorithms for rigid motions (translations and
rotations) are known in this case, see for instance~\cite{leo14}.

In three dimensions, most previous results are concerned with
``regular'' packing problems where objects to be packed are
axis-parallel boxes, see for
instance~\cite{DBLP:journals/jcst/DiedrichHJTT08,DBLP:conf/sofsem/JansenP14}.
In particular, approximation algorithms for packing rectangular
cuboids or convex polyhedra into minimum-volume rectangular cuboids or
convex containers under rigid motions are known~\cite{ISAAC}.  Whether
this is possible under only translations remains open.

In this paper, we provide several results for packing disks of unit
radius under translations into minimum-volume containers of different
types.  In fact, we also handle a higher-dimensional analogue, as
follows.  Let~$d \geq 1$ be fixed. We define a \emph{unit hyperdisk}
to be a $d$-dimensional unit-radius ball in~$\Reals^{d+1}$.  The task
is to pack unit hyperdisks under translations in containers of
dimension~$d+1$ that may be axis-parallel boxes (see
Figure~\ref{fig:exampleDiskPacking} for an example with $d=2$), boxes
that are arbitrarily oriented, or arbitrary convex bodies.  In all
cases, the objective is to minimize the $(d+1)$-volume of the
container.  We describe constant-factor approximation algorithms for
these problems.  Our approximation factors are high and grow quickly
with the dimension, but it is of theoretical interest that the
problem, whose decision version is neither known to be NP-hard nor to
be in NP, can be approximated within a constant factor.  Our
techniques also show that for~$n$ unit hyperdisks, a container of
volume~$O(n^{\frac{d-1}{d}})$ is always sufficient and sometimes
necessary.  As a byproduct, this implies that for~$d \geq 2$ there is
no finite size $(d+1)$-dimensional convex container into which all
unit hyperdisks can be packed simultaneously.

\medskip 

The input to our problem is a set of~$n$ unit hyperdisks
in~$\Reals^{d+1}$.  Nearly all the hyperdisks we consider are unit,
and we often drop the adjective. Since we are allowed to freely
translate the hyperdisks, we can assume each hyperdisk to be centered
at the origin, so that it is fully defined by its normal vector.  Two
hyperdisks are \emph{parallel} if their normal vectors are multiples
of each other (in particular, a vector and its negative define
parallel hyperdisks).  A set of \emph{distinct hyperdisks} is a set
where no two hyperdisks are parallel.

Two hyperdisks \emph{overlap} if there is a point that lies in the
relative interior of both.  If they intersect but do not overlap, then
they \emph{touch}.  By \emph{non-overlapping}, we mean a placement of
hyperdisks such that no two hyperdisks overlap, but where hyperdisks
are allowed to touch.  In other words, we treat hyperdisks as if they
consist of their relative interior only.  With this, we do not need to
consider placements of the hyperdisks where they are arbitrarily close
but remain disjoint.

Our main problem is as follows (see Figure~\ref{fig:exampleDiskPacking}):
\begin{definition}
  Given a set of distinct unit hyperdisks in~$\Reals^{d+1}$, 
  find a container of minimum volume in~${\Reals^{d+1}}$
  such that all hyperdisks can be translated into the container without overlap.
  The problem is called
  \begin{itemize}
  \item \packAA when the container is an axis-parallel box;
  \item \packB when the container is a box of arbitrary orientation;
  \item \packC when the container is a convex body.
  \end{itemize}
  In all cases we also want an actual packing of the hyperdisks inside the container.
\end{definition}

\begin{figure}%
  \centerline{\hfill\includegraphics[scale=0.16]{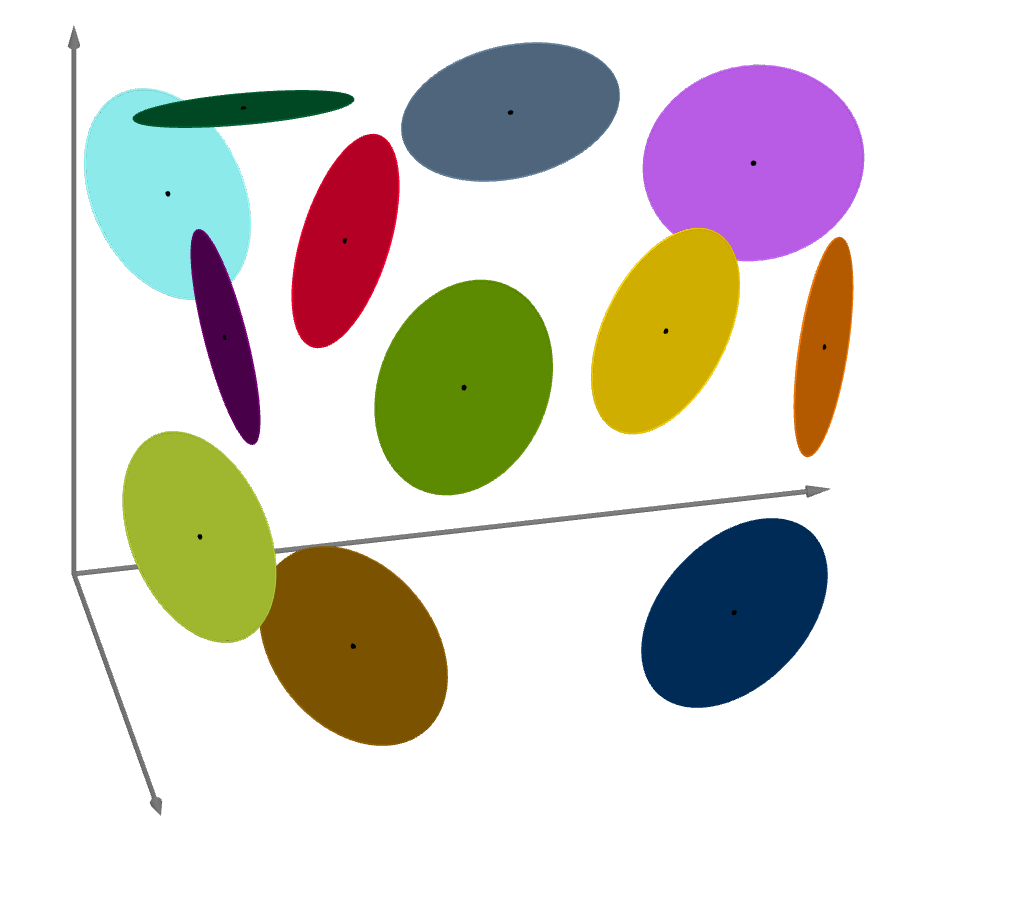}\hfill
    \includegraphics[trim={5cm 2cm 10cm
        0},clip,scale=0.16]{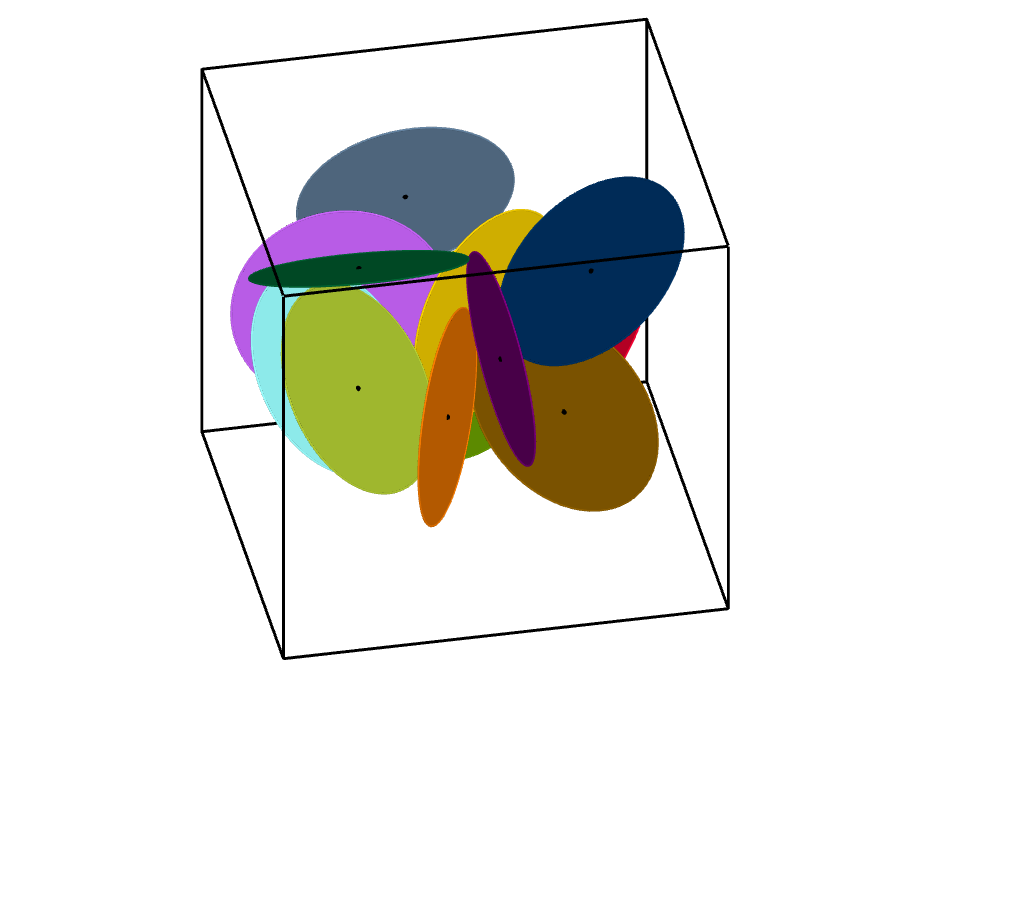}\hfill}
  \caption{Input and solution for the \packAA-problem.}
  \label{fig:exampleDiskPacking}
\end{figure}

For example, the special case of $d=2$ is about packing 2-dimensional
unit disks in $\Reals^3$, while for $d=1$ we talk about packing unit 
segments in the plane.
It is clear that the optimal values of \packAA, \packB and \packC
are monotonically decreasing.

We will approximately solve these packing problems by arranging subsets 
of hyperdisks
such that their centers lie on a line and they are packed as densely
as possible.  Let $O$ be a given ordering of a set of hyperdisks, and
let~$s\in\Reals^{d+1}$ be a vector.  We define the \emph{length}
of~$O$ with respect to direction~$s$ as follows: let~$\ell$ be a line
with direction vector~$s$. Place the hyperdisks such that their
centers lie on~$\ell$ and appear in the ordering~$O$ when traversing
the line in direction~$s$, and such that two consecutive hyperdisks
touch but do not overlap.  Then the \emph{length of the ordering}~$O$
with respect to~$s$ is the distance from the center of the first
hyperdisk to the center of the last hyperdisk.
Figure~\ref{fig:example-stab-2D} illustrates the definition for~$d=1$,
where the hyperdisks are line segments.
\begin{figure}[ht]
  \centerline{\includegraphics{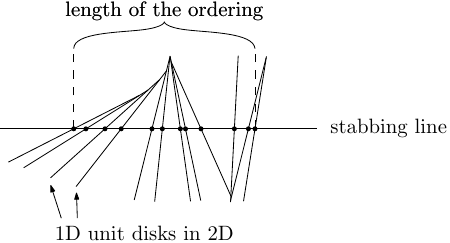}}
  \caption{The length of an ordering of $1$-dimensional hyperdisks.}
  \label{fig:example-stab-2D}
\end{figure}

We can now define the \stab-problem as follows (see 
Figure~\ref{fig:exampleDiskStabbing}):
\begin{definition}[\stab]
  Given a set of distinct unit hyperdisks in~$\Reals^{d+1}$, 
  and an additional vector $s\in\Reals^{d+1}$ defining 
  the direction of a line, find an ordering of the hyperdisks that 
  minimizes the length with respect to~$s$.
\end{definition}

\begin{figure}[ht]
  \centerline{\hfill\includegraphics[scale=0.16]{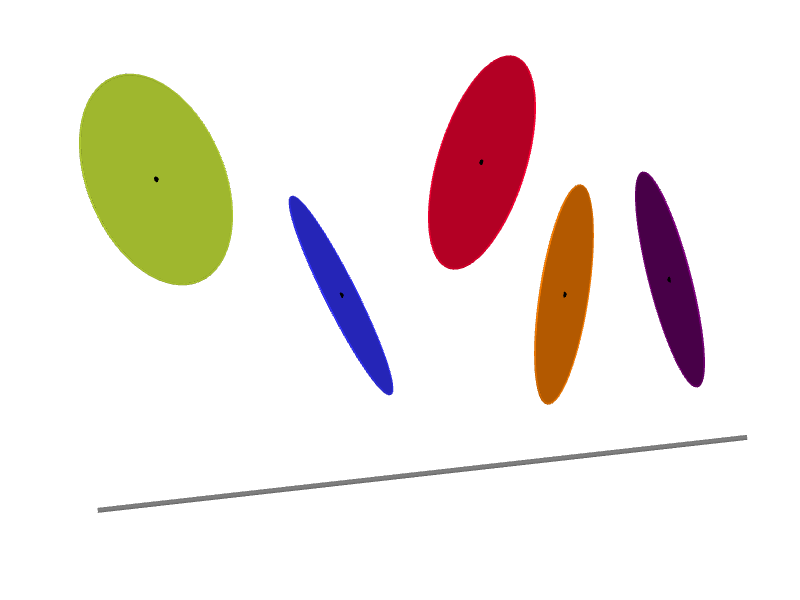}\hfill
    \includegraphics[scale=0.16]{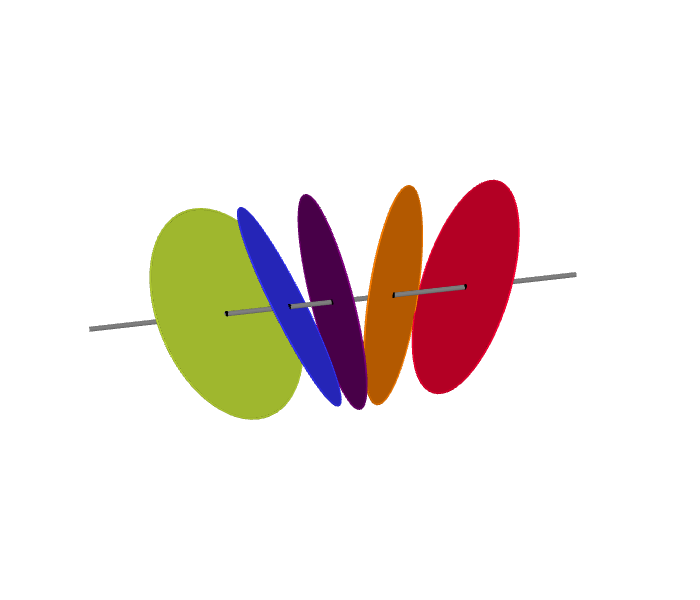}\hfill}
  \caption{Input and solution for the \stab-problem with $d=2$.}
  \label{fig:exampleDiskStabbing}
\end{figure}

In~Section~\ref{sec:metric}, we study the geometry of the
\stab-problem and define a metric on the family of unit hyperdisks.  
An important consequence of this metric is that each solution to the
\stab-problem gives a non-overlapping packing of the unit hyperdisks.  
Note that this is not obvious because the definition of~\stab 
only requires that consecutive hyperdisks in the ordering are 
non-overlapping, but in principle non-consecutive hyperdisks in
the ordering could still overlap. Here it is key that we
restrict our attention to \emph{unit} hyperdisks.
This property will allow us to prove that an optimal solution 
to~\stab implies an approximate solution to~\packAA.  We prove 
this first for hyperdisks with similar normal vectors in 
Section~\ref{sec:similar}, then for the 
general case in Section~\ref{sec:general}.

In Section~\ref{sec:solving} we then show how to approximately
solve \stab by (approximately) computing a shortest Hamiltonian path
in a complete weighted graph, implying a 
constant-factor approximation for~\packAA.

In Section~\ref{sec:convex} we consider the 
\packB and \packC problems, where we reuse several of the ideas
employed for the case of axis-parallel containers.

For $d=1$, all (infinitely many) unit-length line segments can be
packed into a rectangle of area~$2$
(Figure~\ref{fig:segment-packing}), and even smaller containers are
possible. In Section~\ref{sec:unbounded} we show that a similar result
does not hold for $d\geq 2$: there is no bounded-volume convex 
container into which all $d$-dimensional hyperdisks can be packed.
\begin{figure}[htb]
  \centerline{\includegraphics{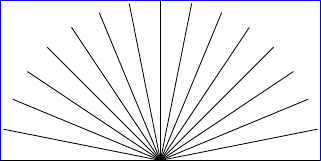}}
  \caption{All unit-length segments can be packed into a container of area~$2$.}
  \label{fig:segment-packing}
\end{figure}
More precisely, we consider a family~$\disks$ of~$n$ hyperdisks
in~$\Reals^{d+1}$, a stabbing direction~$s$, and an angle~$\phi <
\pi/2$, such that the normal of each disk in~$\disks$ makes angle at
most~$\phi$ with~$s$.  We show that there is then a stabbing
of~$\disks$ in direction~$s$ of length~$O(n^{\frac{d-1}{d}})$ (and
therefore also a packing of this volume).  We also construct such a
family~$\disks$ where this bound is tight.

The reader may wish to contrast this with the \emph{Potato Sack
Theorem} of Auerbach et al.~\cite{potato-sack}, see
Kosi\'nski~\cite{Kosinski1957} for a proof and the survey by 
Fejes T{\'o}th~\cite{FejesToth2023} for a general discussion.  
It states that when rigid
motions are allowed, any sequence of convex bodies in~$\Reals^{d}$ can
be packed into a box of finite volume, as long as there is a common
bound on the diameter of the bodies and their total volume is bounded.

In~Section~\ref{sec:OtherObjects} we briefly consider packing other
fixed shapes like squares and mention some open problems.

\section{A metric on unit hyperdisks}
\label{sec:metric}

Given a vector $s \in \mathbb{S}^{d}$, we define the
\emph{$s$-distance}~$\ds(D_1, D_2)$ between two hyperdisks~$D_1$
and~$D_2$ with centers~$c_1$ and~$c_2$ as the length of the
ordering~$D_1, \, D_2$ with respect to~$s$.  In other words, it is the
distance~$|c_1 c_2|$ when the hyperdisks touch and the
ray~$\overrightarrow{c_1 c_2}$ has direction~$s$.  When $D_1$
and~$D_2$ are parallel, we define~$\ds(D_1, D_2) = 0$.

\begin{theorem}
  For any $s \in \mathbb{S}^{d}$, the distance~$\ds$ is a metric on
  the set of $d$-dimensional unit hyperdisks with normals not orthogonal
  to~$s$ (and where parallel hyperdisks are considered equivalent).
  \label{thm:metric}
\end{theorem}
\begin{proof}
  If $D_1$ and $D_2$ are not parallel, it is clear that $\ds(D_1,
  D_2) > 0$. Otherwise, $\ds(D_1, D_2) = 0$ by definition.

  Symmetry holds since a point reflection keeps both hyperdisks
  identical (their normal vectors are negated), but reverses the order
  in which a line with direction~$s$ meets them.

  It remains to prove the triangle inequality~$\ds(D_1, D_3) \leq
  \ds(D_1, D_2) + \ds(D_2, D_3)$.  It clearly holds when any two of
  the three hyperdisks are parallel, so we will assume for a
  contradiction that there are three distinct hyperdisks~$D_1, D_2,
  D_3$ with centers~$c_1, c_2, c_3$ such that $\ds(D_1, D_3) >
  \ds(D_1, D_2) + \ds(D_2, D_3)$.  We place $c_1, c_2, c_3$ in this
  order on a line~$\ell$ with direction~$s$ such that $|c_1 c_2| =
  \ds(D_1, D_2)$ (so the hyperdisks touch but do not overlap) and
  that~$|c_1 c_3| = \ds(D_1, D_3)$ (so these hyperdisks touch in a
  point~$p$ and do not overlap).  Since~$|c_2 c_3| = \ds(D_1, D_3) -
  \ds(D_1, D_2) > \ds(D_2, D_3)$, the hyperdisks~$D_2$ and~$D_3$ do
  not intersect at all.

  We consider the situation in the (two-dimensional) plane~$h$
  containing the line~$\ell$ and the point~$p$.  Since $c_i \in h$ and
  the normal of~$D_i$ is not orthogonal to~$s$, the intersection~$e_i
  := D_i \cap h$ is a segment of length~two centered at~$c_i$, for~$i
  \in \{1, 2, 3\}$.  We choose a coordinate system for~$h$
  where~$\ell$ is the~$x$-axis and~$p$ lies below $\ell$, see
  Figure~\ref{fig:metric}.  By assumption, the segment~$e_2$ has
  length two, its midpoint is~$c_2$, and it does not intersect~$e_1$
  or~$e_3$.

  We first observe that if $\angle pc_1 c_2$ or~$\angle pc_3 c_2$ is
  obtuse, then the entire triangle~$c_1 p c_3$ lies inside a circle of
  radius~$1$ around~$c_2$, so~$e_2$ would have to intersect~$e_1$
  or~$e_3$, a contradiction.  It follows that~$\angle pc_1 c_2$
  and~$\angle pc_3 c_2$ are acute, and~$e_2$ is therefore contained in
  the wedge bounded by the rays~$pc_1$ and~$pc_2$ (light blue in
  Figure~\ref{fig:metric}).
  \begin{figure}[ht]
    \centerline{\includegraphics[page=1]{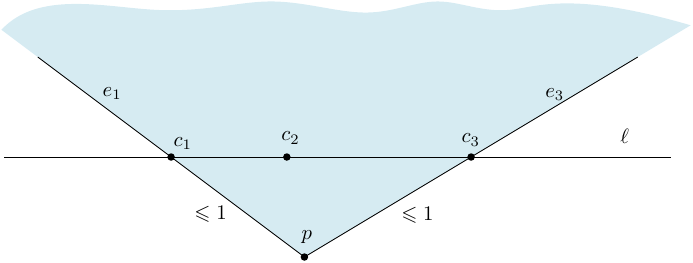}}
    \caption{Proof of Theorem~\ref{thm:metric}, first step.}
    \label{fig:metric}
  \end{figure}

  We can therefore move~$c_1$ to the left and~$c_3$ to the right while
  keeping~$c_2$ and~$p$ unchanged, until~$|c_1 p| = |p c_3| = 1$.
  This can only grow the wedge, so it will still contain~$e_2$.
  The triangle $c_1pc_3$ is isosceles, as shown in  
  Figure~\ref{fig:metric}.
  Let~$u$ be the upper, $v$~the lower endpoint of~$e_2$.  Clearly~$v$
  lies inside the triangle~$c_1 p c_3$.  Breaking symmetry, we assume
  that~$e_2$ has a negative slope, so that~$v$ is the \emph{right}
  endpoint of~$e_2$.
  
  Consider now the segment~$e_2'$ parallel to~$e_2$ centered at~$c_1$,
  with upper endpoint~$u'$ and lower endpoint~$v'$.  The points~$p,
  u', v'$ all lie on the circle~$C$ of radius~$1$ around~$c_1$, see
  Figure~\ref{fig:metric3}.

  Since~$v$ lies above~$p$, so does~$v'$, and therefore~$v'$ lies on
  the counterclockwise arc of~$C$ from~$p$ to~$\ell$.  Let now~$u_1$ be the point
  on~$e_1$ on the horizontal line~$u'u$, and let~$v_i$ be the point
  on~$e_i$ on the horizontal line~$v'v$, for $i \in \{1, 3\}$, see
  Figure~\ref{fig:metric3}.
  \begin{figure}[ht]
    \centerline{\includegraphics[page=3]{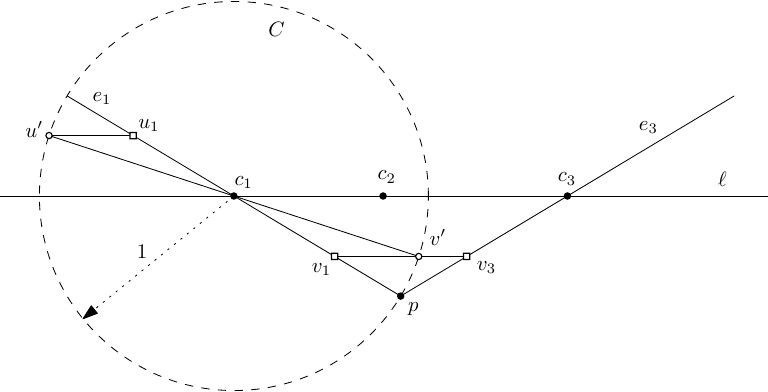}}
    \caption{Proof of Theorem~\ref{thm:metric} (final step).}
    \label{fig:metric3}
  \end{figure}
  Since~$v'$ lies to the right of~$p$, we have~$|v' v_3| < |v' v_1| =
  |u' u_1|$.  However, $|u'u| = |v' v|$, so~$v \in \triangle c_1 p
  c_3$ implies $|v'v| < |v'v_3|$, so $|u'u| < |v'v_3| < |u' u_1|$, and
  $u$~lies outside the wedge bounded by the rays~$pc_1$ and~$pc_2$, a
  contradiction.
\end{proof}

As mentioned in the Introduction, Theorem~\ref{thm:metric} implies
that any solution to the \stab-problem for unit hyperdisks gives
a non-overlapping set of hyperdisks. In fact, this has been the
essence of its proof.

Consider now a configuration of two hyperdisks~$D_1$ and~$D_2$ that
realize the $s$-distance~$\ds(D_1, D_2)$ for a given direction
vector~$s$.  We let~$n_i$ denote the normal of hyperdisk~$D_i$,
$c_i$~its center, and~$h_i$ the hyperplane supporting~$D_i$.  Let~$g$
be the $(d-1)$-dimensional affine subspace~$h_1 \cap h_2$ (for~$d=2$
this is a line with direction vector~$n_1 \times n_2$).  The two 
hyperdisks must touch in a point~$p \in g$.  Let $\hh_i$ denote the
half-hyperplane of~$h_i$ delimited by~$g$ that contains~$c_i$.
Note that the intersection of~$g$ and the hyperdisk~$D_i$ is 
either the point $p$ or a $(d-1)$-dimensional ball of positive radius.

We can distinguish three cases that cover all possibilities (but (i) and (ii)
do not exclude each other):
\begin{itemize}
\item[(i)] $D_1$ lies entirely in~$\hh_1$. In this case
	$D_1$ is tangent to~$g$ in~$p$, $|pc_1|= 1$ and~$|pc_2| \leq 1$;
\item[(ii)] $D_2$ lies entirely in~$\hh_2$. In this case
	$D_2$ is tangent to~$g$ in~$p$, $|pc_2|= 1$ and~$|pc_1| \leq 1$;
\item[(iii)] $D_1$ does not lie entirely in~$\hh_1$ and
    $D_2$ does not lie entirely in~$\hh_2$.
    In this case $|c_1p| = |pc_2| = 1$ and~$g$ intersects both hyperdisks
    in $(d-1)$-dimensional balls of positive radius that have disjoint 
	interiors, and therefore are tangent at $p$.
\end{itemize}
The last case may be harder to visualize. In $\Reals^{d+1}$ such case
can be constructed as follows, see Figure~\ref{fig:caseiii} for $d=2$. 
Take two hyperdisks in~$\Reals^d$ that intersect, let~$p$ be a point 
where the boundaries of the hyperdisks intersect, let~$g$ be a 
$(d-1)$-dimensional affine subspace through $p$ that does not contain
other intersection points of the hyperdisks and that goes through the 
interior of both hyperdisks, and rotate one of the hyperdisks around $g$.
The point $p$ remains the touching point of the hyperdisks and
$g$ intersects the interior of both hyperdisks.
\begin{figure}[ht]
  \centerline{\includegraphics[page=1]{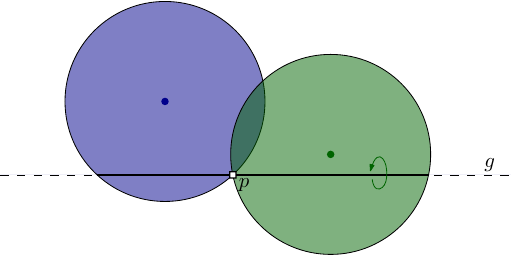}}
  \caption{Example showing case~(iii) for $d=2$, which is embedded 
	in~$\Reals^3$. The disk on the right is rotated around 
	the line $g$ in $\Reals^3$ and then both disks touch only at $p$.}
  \label{fig:caseiii}
\end{figure}
  
We have the following lemmas.
\begin{lemma}
  \label{lem:sinxi}
  Let~$D_1$ and~$D_2$ be two unit hyperdisks whose normal vectors form an
  angle~$\xi$. Then~$\ds(D_1,D_2) \geq \sin\xi$.
\end{lemma}
\begin{proof}
  The angle formed by~$\hh_1$ and~$\hh_2$ is either~$\xi$
  or~$\pi-\xi$.  By possibly negating~$n_2$, we can assume the former.
  The claim is then obvious in case~(i), as~$c_1$ has
  distance~$\sin\xi$ from the hyperplane~$h_2$, and case~(ii) is
  symmetric.

  In case~(iii) we consider the isosceles triangle~$\triangle c_1pc_2$ 
  and set~$\gamma := \measuredangle c_1 p c_2$.  
  We will show~$\gamma \geq \xi$, so we
  have $|c_1 c_2| = 2 \sin(\gamma/2) \geq 2\sin(\xi/2) \geq \sin\xi.$

  To see~$\gamma \geq \xi$, let~$g^{\perp}$ be the $2$-dimensional 
  affine subspace orthogonal to $g$ and through~$p$.
  Let~$c'_i$ be the orthogonal projection of~$c_i$ on~$g$, 
  and let~$c^{\perp}_i$ be the orthogonal projection of $c_i$ on $g^{\perp}$. 
  See Figure~\ref{fig:caseiii2}.  Because we are using orthogonal projections
  on orthogonal affine subspaces through $p$, 
  we have $c_i-c^\perp_i=c'_i-p$. Moreover, the vectors $(c_i-c'_i)$ 
  and $(c'_j-p)$ are orthogonal for $i,j\in \{1,2\}$.
  We will also use that the orthogonal projection of a segment cannot 
  be longer than the original segment. Therefore 
  $|c^\perp_i p|\le |c_i p|=1$ and $|c'_1c'_2|\le |c_1 c_2|$.
  \begin{figure}[ht]
    \centerline{\includegraphics[page=2]{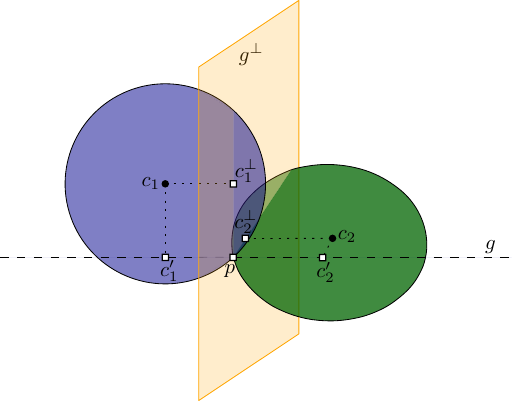}}
    \caption{Proof of Lemma~\ref{lem:sinxi} for case~(iii).}
    \label{fig:caseiii2}
  \end{figure}

  Since $D_1\cap g$ and $D_2\cap g$ are interior disjoint 
  $(d-1)$-dimensional balls of positive radius tangent at $p$ 
  and contained in the $(d-1)$-dimensional subspace $g$, 
  their centers have to be collinear with $p$.
  Noting that $c'_i$ is the center of $D_i\cap g$,
  we get that the points $c'_1$, $p$ and $c'_2$ are collinear
  and in that order. Therefore $(c'_1-p)\cdot (c'_2-p)<0$.
 
  If $c'_i=c_i$, then $p=c^\perp_i$ and
  thus $|c_1c_2|\ge |c'_1 c'_2| \ge |p c'_i|=1 \ge \sin \xi$.
  Therefore, we consider only the case where~$c'_1\neq c_1$ and 
  $c'_2\neq c_2$, which means that $p\neq c^\perp_1$ and $p\neq c^\perp_2$.
  We then have $\measuredangle c^\perp_1 p c^\perp_2 = \xi$
  because rotating $D_2$ an angle of $\xi$ around $g$ would place
  $D_2$ in the hyperplane supporting $D_1$.
  We have
  \[
	(c^\perp_1 - p) \cdot (c^\perp_2 - p) = 
       |c^\perp_1 p|\cdot |c^\perp_2 p| \cdot \cos \xi \le 
	   |c_1 p|\cdot |c_2 p|\cdot \cos \xi = \cos \xi,
  \]
  and therefore   
  \begin{align*}
    \cos\gamma & = (c_1 - p)\cdot(c_2 - p) \\
    & = \big((c_1 - c'_1) + (c'_1 - p)\big)\cdot\big((c_2 - c'_2) + (c'_2 - p)\big) \\
    & = (c_1 - c'_1)\cdot (c_2 - c'_2)
    + (c'_1 - p)\cdot (c'_2 - p) \\
    & < (c^\perp_1 - p)\cdot (c^\perp_2 - p) + 0 \\
	& \leq \cos\xi. \qedhere
  \end{align*}
\end{proof}

\begin{lemma}
  \label{lem:upsinxi}
  Let~$D_1$ and~$D_2$ be two unit hyperdisks whose normal vectors form an
  angle~$\xi$, and such that their normals make an angle of at
  most~$\phi < \pi/2$ with the direction~$s$. Then~$\ds(D_1,D_2) \leq
  \frac{\sin\xi}{\cos \phi}$.
\end{lemma}
\begin{proof}
  Again we can assume that the angle between~$\hh_1$ and~$\hh_2$
  is~$\xi$.  This means that~$c_2$ has distance~$t \leq \sin\xi$
  from~$h_1$.  Let~$c'_2$ be the point at distance~$t$ on the ray
  from~$c_1$ with direction~$n_1$. Since the ray~$c_1 c_2$ has
  direction~$s$, then~$\measuredangle c'_2 c_1 c_2
  \leq \phi$ and it follows that~$t/|c_1 c_2| = 
  \cos (\measuredangle c'_2 c_1 c_2) \geq \cos\phi$.
  This implies~$|c_1 c_2| \leq t/\cos\phi$ and the claim follows.
\end{proof}

\section{Stabbing helps to pack similar unit hyperdisks}
\label{sec:similar}

We fix the angle~$\phi_0 = \arccos(1/\sqrt{d+1})$. (For $d = 2$,
$\phi_0 \approx 54.7^{\circ}$.) This is the angle between the diagonal
and an edge of the $(d+1)$-dimensional hypercube.  In this section we
consider a set~$\disks$ of unit hyperdisks whose normal vector makes an
angle of at most~$\phi_0$ with the positive~$x_{d+1}$-axis, that is,
with the vector~$(0, 0, \dots, 0, 1)$.  We will show that the \stab
and \packAA problems are related by proving that an optimal solution
to~\stab for~$\disks$ with respect to direction~$s = (0, 0, \dots, 0,
1)$ provides a constant-factor approximation to~\packAA for~$\disks$.
In the next section we will then extend this to arbitrary sets of unit
hyperdisks.

Let us first define~$E$ to be the maximum \emph{extent} of all the
hyperdisks.  Formally, let~$E$ be the smallest axis-parallel box that
contains all the hyperdisks if we place them with their center at the
origin, and let~$E_1, E_2, \dots, E_{d+1}$ be the dimensions of~$E$.
We have 
\[
2 \geq E_i \geq 2 \sin(\pi/2 - \phi_0) = 2 \cos\phi_0 = 2/
\sqrt{d+1}, ~~~~~\text{for~$1 \leq i \leq d$}. 
\]

Let now~$M$ denote a minimum-volume axis-parallel box into
which~$\disks$ can be packed---that is, an optimal solution to
\packAA---and denote its dimensions by~$M_1, M_2, \dots, M_{d+1}$.  Its
volume is~$\OPT = M_1 M_2 \dots M_{d+1}$.

When we project~$M$ onto a hyperplane normal to the $x_{d+1}$-axis, we
obtain a $d$-dimensional box of size~$M_1 \times M_2 \times \dots
\times M_d$.  We place a grid of $n_1 \times n_2 \times \dots \times
n_d$~points inside this box with grid size~$\mu = 1/\sqrt{d(d+1)}$,
starting at the point at distance~$1/{\sqrt{d+1}} + {\mu}/{2}$ from
the boundary, and stopping when we have passed the same distance from
the opposite boundary. See the left side of Figure~\ref{fig:gridstab}.  
We have
\[
n_i = \left\lceil \frac{M_i - \frac 2{\sqrt{d+1}}}{\mu} \right\rceil <
\frac{M_i - \frac 2{\sqrt{d+1}} + \mu}{\mu} < \frac{M_i}\mu, ~~~~~
 \text{for~$1 \leq i \leq d$.}
\]
\begin{figure}[ht]
  \centerline{\includegraphics[width=\textwidth]{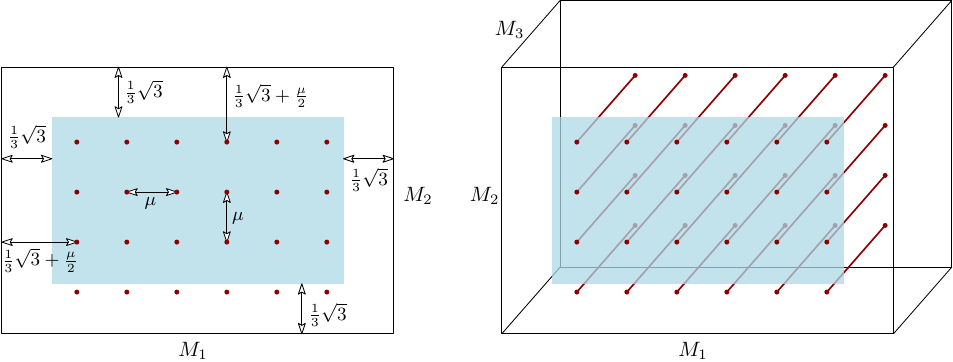}}
  \caption{Left: Placing a grid on a $d$-dimensional face of the optimal 
		container~$M$, for~$d=2$. 
		Right: the vertical segments stabbing the hyperdisks.}
  \label{fig:gridstab}
\end{figure}

We observe that since the hyperdisks lie entirely in~$M$, the
projection of a hyperdisk center on the box must have distance at
least~$1/\sqrt{d+1}$ from the box boundary (that is, it must lie in
the shaded region of Figure~\ref{fig:gridstab}), so it has distance at
most~$\frac{\mu}{2}\sqrt{d}$ to the nearest grid point.

For each of the~$n_1 \times \cdots \times n_d$ grid points, we
consider the \emph{vertical} segment of length~$M_{d+1}$ obtained by intersecting~$M$
with a line parallel to the~$x_{d+1}$-axis through the grid point.  
See the right side of Figure~\ref{fig:gridstab}.  By
our construction, every disk center has distance at
most~$\frac\mu2\sqrt{d}$ to one of these segments.  This means that
such a segment stabs the hyperdisk in a point~$q$ at distance at
most~$\frac\mu2\frac{\sqrt{d}}{\sin(\pi/2 - \phi_0)} = \frac\mu2\sqrt{
  d(d+1)}= \frac 12$ from the disk center.  The smaller disk of
radius~$1/2$ around~$q$ lies entirely inside the original disk, and
we replace each original disk with this smaller disk.
With this, we have obtained $n$ hyperdisks of radius $1/2$ that
are packed inside the same box and such that each of them is pierced 
through the center by one of the vertical segments through the grid points.

We can now place all $n_1 n_2 \cdots n_d$ vertical segments behind
each other and obtain a hyperdisk stabbing of length~$n_1 n_2 \cdots
n_d M_{d+1}$ for~$\disks$, where the hyperdisks have now radius~$1/2$. 
Simply enlarging this stabbing by factor~$2$ results
in a solution to~\stab for~$\disks$ of length
\begin{align*}
  2 \cdot n_1 n_2 \cdots n_d \cdot M_{d+1} &
  < 2 \cdot \frac {M_1}{\mu} \frac{M_2}{\mu}
  \cdots \frac{M_d}{\mu} \cdot M_{d+1} \\
  & = \frac{2}{\mu^d} \cdot M_1 M_2 \cdots M_{d+1} = \frac{2}{\mu^d}
  \OPT \\
  & = 2\sqrt{\big(d(d+1)\big)^{d}} \OPT
  < 2 (d+1)^{d} \OPT.
\end{align*}

\begin{theorem}
  Let~$\disks$ be a family of unit hyperdisks whose normals make an angle
  of at most~$\phi_0 = \arccos(1/\sqrt{d+1})$ with
  the~$x_{d+1}$-axis. An optimal solution to~\stab for~$\disks$ with
  respect to the~$x_{d+1}$-direction has length at
  most~$2(d+1)^{d}\cdot \OPT$, where~$\OPT$ is the volume of an
  optimal solution to~\packAA for~$\disks$. The \stab-solution thus
  provides a $((2^{d+1} (d+1)^{d} + 1))$-approximation for~\packAA.
  \label{thm:similar}
\end{theorem}
\begin{proof}
  Since an optimal solution to~\packAA implies a stabbing of length at
  most~$2(d+1)^{d}\cdot \OPT$, an optimal stabbing has at most this
  length, implying the first claim.

  The smallest axis-parallel box enclosing the stabbing is a valid
  packing of the hyperdisks. Here we make use of the triangle
  inequality from Theorem~\ref{thm:metric}: if consecutive hyperdisks
  do not overlap, then all hyperdisks are non-overlapping.
  If~$L_{d+1}$ is the length of the stabbing, its volume is at
  most~$E_1 E_2\cdots E_d (L_{d+1} + E_{d+1})$ (the $E_{d+1}$-term is
  needed because we measure the length of a stabbing only from the
  first to the last center).  Since $E_i \leq 2$ for $1\leq i\leq n$ 
  and $L_{d+1} \leq 2(d+1)^{d}\cdot \OPT$, this is at most
  \[
  2^{d} \cdot 2(d+1)^{d} \cdot \OPT + E_1 E_2 \cdots E_{d+1} \leq
  (2^{d+1} (d+1)^{d} + 1) \cdot 
  \OPT ,
  \]
  where we made use of the fact that~$\OPT \geq E_1 E_2 \cdots E_{d+1}$.
\end{proof}

\section{Stabbing helps to pack general unit hyperdisks inside an axis-parallel box}
\label{sec:general}

We now address the general \packAA problem. Given a set~$\disks$
of unit hyperdisks, we partition it into $d+1$~subsets~$\disks_1,
\disks_2, \dots, \disks_{d+1}$, depending on which coordinate axis the
normal vector makes the smallest angle with (this corresponds simply
to the highest absolute coordinate value in the normal vector).  We
then consider an optimal stabbing for~$\disks_i$ with respect to
the~$x_i$-axis, for $1 \leq i \leq d+1$. Note that the normals for
each subfamily make an angle of at most~$\phi_0$ with its
corresponding axis.

\begin{theorem}
  Optimal stabbings for~$\disks_i$ with respect to the $x_i$-axis, for
  $1 \leq i \leq d+1$, can be used to provide a
  $3\big(3d+3\big)^{d+1}$-approximation for~\packAA for~$\disks$.
  \label{thm:general}
\end{theorem}
\begin{proof}
  Again we define~$E$ (of dimension~$E_1 \times E_2 \times \dots
  \times E_{d+1}$) to be the \emph{extent} of all the hyperdisks (in the
  entire set~$\disks$).  If $\disks_i = \emptyset$ for all but one
  index~$i$, then we are done by Theorem~\ref{thm:similar}, so we
  assume this is not the case.  This implies that for all extent
  dimensions we have
  \[
  2/\sqrt{d+1} \leq E_i \leq 2.
  \]
  Since $\OPT \geq E_1 E_2 \cdots E_{d+1}$, we have~$\OPT \geq
  2^{d+1}/(d+1)^{(d+1)/2}$, and therefore $(d+1)^{d}\OPT \geq 1$.

  Consider now an index~$i \in \{1,2,\dots,d+1\}$ where~$\disks_i \neq
  \emptyset$. Let~$\ell^*$ be the length of an optimal stabbing for the
  set~$\disks_i$ in the direction of the~$x_i$-axis, and let~$s$ be a
  segment of length~$\ell^*$ parallel to the~$x_i$-axis that stabs the
  hyperdisks in~$\disks_i$.  We partition~$s$ into $\lceil \ell^* \rceil$
  pieces of length at most one.  For each piece, we take the smallest
  axis-parallel box containing the hyperdisks whose centers lie on the
  piece. Each box has dimensions at most~$E_1 \times \dots \times
  E_{i-1} \times (1 + E_i) \times E_{i+1} \times \dots \times
  E_{d+1}$, so it fits into a hypercube of side length~$3$.
  The number of these hypercubes needed for~$\disks_i$ is
  \[
  \lceil \ell^* \rceil \leq \ell^* + 1 
  \leq 2(d+1)^{d}\cdot \OPT + 1
  \leq 3(d+1)^{d}\cdot \OPT.
  \]

  Repeating this for all~$d+1$ sets~$D_i$, we find that we can
  pack~$\disks$ into
  \[
  (d+1)\cdot 3 \cdot (d+1)^{d} \cdot \OPT = 3\cdot(d+1)^{d+1} \cdot \OPT
  \]
  hypercubes of side length~3, with a total volume of
  \[
  3^{d+1} \cdot 3\cdot(d+1)^{d+1} \cdot \OPT =
  3\cdot\big(3d+3\big)^{d+1} \cdot \OPT.
  \]
  We can stack the hypercubes along any coordinate direction to obtain
  an axis-parallel box packing~$\disks$ of the same volume.  Note
  that, at the cost of an additional constant factor of, say,
  $2^{d+1}$, we could even pack~$\disks$ into a hypercube. This is for
  example obtained by doubling each dimension of the hypercube
  container whenever it becomes too small to continue placing the
  cubes of side length~$3$.
\end{proof}

\section{Approximating the optimal stabbing and packing inside an axis-parallel box}
\label{sec:solving}

It remains to show how to actually solve the \stab problem.
\begin{theorem}
  A $\frac 32$-approximation for \stab can be computed in polynomial time.
  \label{thm:approxStab}
\end{theorem}
\begin{proof}
  Given a set of~$n$ unit hyperdisks by their normal vectors
  and a direction vector~$s$, we construct the complete weighted
  graph~$G$ on~$n$ vertices~$\{1, 2, \dots, n\}$, where the weight of
  the edge~$(i, j)$ is~$\ds(D_i, D_j)$.  The Hamiltonian paths in this
  graph correspond one-to-one to the orderings of the hyperdisks, and
  the length of the path in the graph is equal to the length of the
  ordering with respect to~$s$.  It follows that the optimal ordering
  is given by the shortest Hamiltonian path in~$G$.

  Christofides' approximation algorithm applies to our problem, since
  by Theorem~\ref{thm:metric} the triangle inequality holds in~$G$.
  We can therefore compute a Hamiltonian path for~$G$ whose length is
  at most $3/2$~times the optimal~\cite{Hoogeveen}.
\end{proof}

We then have the following corollary.
\begin{corollary}
\label{co:general}
  A $4\big(3d+3\big)^{d+1}$-approximation to \packAA can be
  computed in polynomial time.
\end{corollary}
\begin{proof}
  In the proof of Theorem~\ref{thm:general}, we are now using a segment
  of length~$\ell \leq \tfrac{3}{2}\ell^* \leq 3(d+1)^{d}\cdot \OPT$, so $\lceil\ell\rceil \leq
  4(d+1)^{d}\cdot \OPT$, and the number of hypercubes of side
  length~$3$ packing~$\disks$ is at most
  \[
  (d+1)4(d+1)^{d} \cdot \OPT = 4 (d+1)^{d+1} \cdot \OPT,
  \]
  with a total volume of $4\cdot\big(3d+3\big)^{d+1} \cdot \OPT$.
\end{proof}

\section{Packing inside a convex container or a box with arbitrary orientation}
\label{sec:convex}

We now turn our attention to packing unit hyperdisks under
translations inside a convex container of minimum volume or a box of
minimum volume, where we can freely select the orientation of the box.
Again, our aim is to get a constant-factor approximation.

\subsection{Approximating a convex body by a box}

We will need a geometric lemma that shows that finding a box of
(approximately) minimum volume where we can pack the hyperdisks
suffices to get a constant-factor approximation to the (approximately)
smallest convex container.  We phrase the lemma in~$\Reals^d$, even
though later we will use it in~$\Reals^{d+1}$.  Similar arguments have
been used in the literature before, we include the lemma for
completeness.  The factor~$d^{3d/2}$ is not tight, in two dimensions,
for instance, it is known that factor~$2$ is
possible~\cite{SCHWARZKOPF199877}.

\begin{lemma}
  \label{lem:convexbox}
  Any convex body~$K$ in~$\Reals^d$ 
  is contained in a box of volume at most $d^{3d/2}\cdot \vol(K)$.
\end{lemma}
\begin{proof}
  Let~$B(r)$ denote the ball in~$\Reals^d$ centered at the origin 
  with radius~$r$.
  
  It is a classical result in convexity that there is a unique
  ellipsoid~$E$ of maximum volume contained in~$K$. Moreover,
  scaling~$E$ around its center by~$d$ results in an ellipsoid that
  contains~$K$.  This ellipsoid is called the \emph{L\"owner-John
  ellipsoid}; see for example the books by Barvinok~\cite[Chapter~5,
    Section~2]{barvinokbook} or Ben-Tal and
  Nemirovski~\cite[Theorem~4.9.1]{BTN01}.  Therefore, we have $\vol(E)
  \le \vol(K) \le d^d \vol(E)$.  We choose a coordinate system with
  origin at the center of~$E$ and aligned with the axes of~$E$.
  Scaling along each coordinate axis, we can transform~$E$ into the
  unit-radius ball~$B(1)$.  Let~$K'$ be the image of~$K$ under this
  scaling.
  
  We then have $B(1)\subseteq K' \subseteq B(d)$.  Inside the
  ball~$B(1)$ we inscribe the axis-parallel cube~$Q_{\rm in}$ of side
  length~$2/\sqrt{d}$, around the ball~$B(d)$ we circumscribe the
  axis-parallel cube~$Q_{\rm out}$ of side length~$2d$.  Therefore, we
  have obtained $Q_{\rm in}\subseteq K' \subseteq Q_{\rm out}$
  such that
  \[
    \vol(Q_{\rm out}) = (2d)^{d} = 
    d^{3d/2} \cdot \left(\frac{2}{\sqrt{d}}\right)^{d} = 
    d^{3d/2} \cdot \vol(Q_{\rm in})
	\leq d^{3d/2} \cdot \vol(K').
  \]
  
  Applying the inverse scaling along each coordinate axis, such
  that~$B(1)$ is transformed back into~$E$, we obtain the desired box
  because an affine transformation keeps ratios of volumes invariant.
\end{proof}

\subsection{The algorithm}

We will need another lemma that allows us to replace any box by an
axis-parallel box, assuming only that the original box is not too
thin.  We assume~$d \geq 2$.
\begin{lemma}
  \label{lem:repack}
  Assume there is a box~$K$ of arbitrary orientation where~$\disks$
  can be packed, and assume that each side of~$K$ has length at
  least~$1/2$.  Then there is an axis-parallel box~$K'$ where~$\disks$
  can be packed with
  \[
  \vol(K') \leq (2d + 4)^{d+1} \cdot \vol(K).
  \]
\end{lemma}
\begin{proof}
  We cover~$K$ with hypercubes of side length~$1$ that are aligned
  with~$K$.  The total number of these hypercubes is at most
  $2^{d+1}\vol(K)$, since we need to round up each side length to the
  nearest integer, at most doubling the side length.

  Consider now one such small hypercube~$Q$. Its diagonal has
  length~$\sqrt{d+1}$, therefore the projection of~$Q$ on any axis has
  length at most~$\sqrt{d+1} < d$.  We can therefore enclose~$Q$ in an
  axis-parallel hypercube~$Q'$ of side length~$d+2$, such that every
  unit hyperdisk whose center lies in~$Q$ is completely
  contained in~$Q'$.

  It follows that~$\disks$ can be packed into the union of the~$Q'$.
  The total volume of all~$Q'$ is
  \[
  (d+2)^{d+1}\cdot2^{d+1} \cdot\vol(K) = (2d + 4)^{d+1} \vol(K),
  \]
  and we can stack the~$Q'$ along any coordinate direction to obtain
  an axis-parallel box of this volume.
\end{proof}

\begin{theorem}
  \label{thm:arbitrarybox}
  A $4\big(3d+3\big)^{2d+2}$-approximation to \packB 
  can be computed in polynomial time.
\end{theorem}
\begin{proof}
  We select an arbitrary disk~$D_0$ from~$\disks$ and  choose a
  coordinate system where the normal of~$D_0$ is the $x_{d+1}$-axis.
  We then use the algorithm of Corollary~\ref{co:general} to compute
  an axis-parallel box with respect to this coordinate system
  where~$\disks$ can be packed.

  Let~$K_0$ be the box computed by the algorithm, and let~$K^*$ be a
  minimum-volume box (in an arbitrary orientation) where~$\disks$ can
  be packed. Let~$\OPT_B = \vol(K^*)$ be the optimal volume.

  If each side of~$K^*$ has length at least~$1/2$, then by
  Lemma~\ref{lem:repack} there is an axis-parallel (with respect to
  our coordinate system based on~$D_0$) box~$K'$ of
  volume~$(2d+4)^{d+1} \cdot \OPT_B$ where~$\disks$ can be packed.
  Corollary~\ref{co:general} guarantees that
  \[
  \vol(K_0) \leq 4(3d+3)^{d+1} \vol(K')
  \leq 4(3d+3)^{d+1} \cdot (2d+4)^{d+1} \cdot \OPT_B
  < 4(3d+3)^{2d+2} \cdot\OPT_B.
  \]

  It remains to consider the case where~$K^*$ has a side of length
  smaller than~$1/2$.  We will denote both the direction and the
  length of this side by~$s^*$.  From~$s^* < 1/2$ follows that the normal
  of every hyperdisk in~$\disks$ makes an angle at most~$\theta =
  \arcsin \frac{s^*}{2}$ with the direction~$s^*$.  We have~$\theta <
  \arcsin\frac14 \approx 14.48^{\circ} < \phi_{0}/2$.  (Recall that we
  defined the angle~$\phi_0 = \arccos(1/\sqrt{d+1}) \geq
  \arccos(1/\sqrt{3}) \approx 54.74^{\circ}$.)

  Let now~$s_0$ be the direction of the normal of~$D_0$, that is,
  the~$x_{d+1}$-axis in our chosen coordinate system.  Then the normal
  of any hyperdisk in~$\disks$ makes an angle of at most~$2\theta <
  \phi_0$ with~$s_0$.  It follows that the algorithm of
  Corollary~\ref{co:general} will compute a single disk stabbing for
  all disks in~$\disks$ with respect to the direction~$s_0$, and will
  then fit a box to this disk stabbing.

  Consider again direction~$s^*$.  Since the normal of every disk makes
  an angle of at most~$\theta < \phi_{0}$ with~$s^*$, using a suitable
  coordinate system and Theorem~\ref{thm:similar} implies that an
  optimal disk stabbing with direction~$s^*$ has length at
  most~$2(d+1)^d \OPT_B$.  Let $D_1,\dots,D_n$ be the ordering of such
  an optimal disk stabbing for direction~$s^*$.  ($D_0$ is one of these
  disks, possibly in the middle.)  By
  Lemma~\ref{lem:sinxi}, we have $d_{s^*}(D_i, D_{i+1}) \geq \sin\xi_i$,
  where~$\xi_i$ is the angle between the normals of~$D_i$
  and~$D_{i+1}$.  On the other hand, by Lemma~\ref{lem:upsinxi}, we
  have $d_{s_0}(D_i, D_{i+1}) \leq \sin\xi_{i}/\cos(2\theta) \leq
  d_{s^*}(D_i, D_{i+1})/\cos(2\theta)$.  Since $\cos(2\theta) =
  1-2\sin^2\theta = 1 - 2(\frac{s^*}{2})^2 \geq \frac78$, we
  have~$d_{s_0}(D_i, D_{i+1}) \leq \frac 87 d_{s^*}(D_i, D_{i+1})$, so
  the ordering~$D_1, \dots, D_{n}$ gives us a disk stabbing for
  direction~$s_0$ of length at most~$\frac{16}{7}(d+1)^d \OPT_B$.

  This implies that the algorithm, which computes a
  $3/2$-approximation to the optimal disk stabbing in direction~$s_0$,
  returns a disk stabbing of length at most~$\frac{24}{7}(d+1)^d
  \OPT_B \leq 4(d+1)^d \OPT_B$.  The box computed by the algorithm
  therefore has volume at most
  \[
  \vol(K_0) \leq 2^d 4(d+1)^d \OPT_B + \vol(E_0),
  \]
  where~$E_0$ is the extent of~$\disks$ in our chosen
  coordinate-system.  The sides of~$E_0$ orthogonal to~$s_0$ have
  length at most~$2$, the side of~$E_0$ in direction~$s_0$ has length
  at most~$2\sin(2\theta)$ (it is determined by a disk whose normal
  makes the largest angle with direction~$s_0$), so
  \[
  \vol(E_0) \leq 2^{d+1} \sin2\theta
  =  2^{d+1} \cdot 2 \sin\theta \cos\theta
  =  2^{d+1} \cdot s^* \cos\theta.
  \]
  On the other hand, $K^*$ has one side of length~$s^*$, and all other
  sides have length at least~$2\cos\theta = 2\cos(\arcsin \frac{s^*}{2}) >
  2\cos(\arcsin \frac14) = \frac12 \sqrt{15} > 1$. Thus
  \[
  \vol(K^*) \geq (2 \cos\theta)^{d} \cdot s^*
  = 2\cos\theta (2 \cos\theta)^{d-1} \cdot s^*
  > 2s^*\cos\theta \geq \frac{\vol(E_0)}{2^{d}}.
  \]
  It follows that~$\vol(E_0) \leq 2^{d} \cdot\OPT_B$, 
  and we finally have
  \[
  \vol(K_0) \leq 2^d (4(d+1)^d + 1) \cdot \OPT_B
  < 4(3d+3)^{2d+2} \cdot \OPT_B. \qedhere
  \]
\end{proof}

By Lemma~\ref{lem:convexbox}, the smallest box of arbitrary
orientation where we can pack~$\disks$ is a
$(d+1)^{(3d+3)/2}$-approximation to the smallest convex body where we
can pack~$\disks$. Thus, by returning the box computed in
Theorem~\ref{thm:arbitrarybox} we obtain the following.
\begin{corollary}
  A $36\cdot9^d(d+1)^{(7d+7)/2}$-approximation to \packC can be
  computed in polynomial time.
\end{corollary}

\section{The volume needed to stab or pack $n$~hyperdisks}
\label{sec:unbounded}

While in two dimensions all (infinitely many) unit-length line
segments can be packed into a single fixed rectangle
(Figure~\ref{fig:segment-packing}), this is not true in higher
dimensions.  We will show that volume~$O(n^{\frac{d-1}{d}})$ is
sometimes needed but also always sufficient to stab~$n$ similar
unit hyperdisks in~$\Reals^{d+1}$.  This will imply
similar lower bounds for packing inside a convex container
and similar upper bounds for packing inside an axis-parallel box.

\begin{theorem}
  There is a family~$\disks$ of~$n$ unit hyperdisks in~$\Reals^{d+1}$ whose
  normals make an angle of at most~$\phi_0 = \arccos(1/\sqrt{d+1})$
  with the positive~$x_{d+1}$-axis, such that,
  for each subfamily~$\disks'\subset \disks$, any stabbing of~$\disks'$
  in any direction~$s$ needs length~$\Omega(|\disks'|/n^{1/d})$.
  \label{thm:unbounded}
\end{theorem}
\begin{proof}
  The claim is obviously true for~$d = 1$, so we assume~$d \geq 2$.
  We construct a set~$\disks$ of~$n$ hyperdisks as follows: we pick a
  $d$-dimensional hypercube with side length~$c$ in the hyperplane
  normal to the~$x_{d+1}$-axis that is centered at the origin, for
  some constant~$c > 0$. We partition the hypercube into grid cells
  with side length~$\eps = c/n^{1/d}$, and project each grid point in
  the positive $x_{d+1}$-direction onto the unit-radius sphere, see
  Figure~\ref{fig:GridOnSphere}.  These~$n$ points define the normal
  vectors for our set~$\disks$ of hyperdisks.  When~$c$ is small
  enough (depending on $d$), all normals make an angle of less 
  than~$\phi_0$ with the~$x_{d+1}$-axis.
  \begin{figure}[ht]
    \centerline{\includegraphics{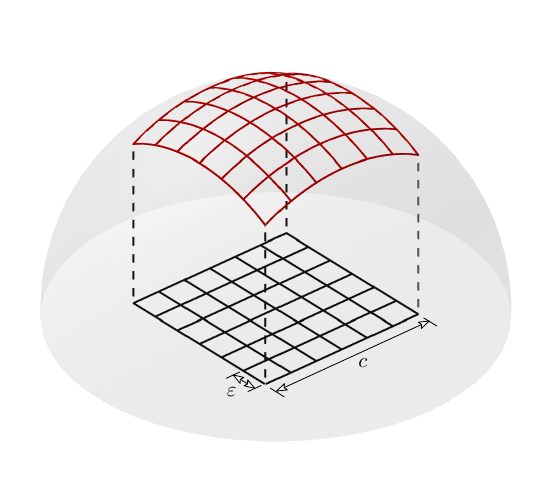}}
    \caption{Projecting a grid onto the unit sphere.}
    \label{fig:GridOnSphere}
  \end{figure}

  We observe next that the angle made by any two of the normal vectors
  is at least~$\eps$. Indeed, since any two grid points~$p_1, p_2$
  have distance at least~$\eps$, so do their projections~$p'_1$
  and~$p'_2$. The angle~$\measuredangle p'_1 o p'_2$ is the length of
  the great circle arc between~$p'_1$ and~$p'_2$ on~$\mathbb{S}^{d}$,
  so it is longer than~$|p'_1p'_2|$, see Figure~\ref{fig:proj}.
  \begin{figure}[ht]
    \centerline{\includegraphics{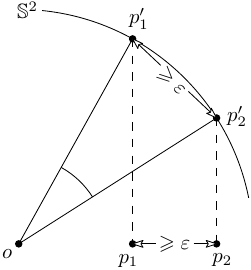}}
    \caption{Two grid points $p_1,p_2$ and their projection onto $\mathbb{S}^2$.}
    \label{fig:proj}
  \end{figure}

  Using Lemma~\ref{lem:sinxi} and the inequality $\sin x\ge x/2$ 
  for all~$x\in [0,\pi/2]$, this implies that for any two 
  hyperdisks~$D_1, D_2 \in \disks$ and any direction~$s$ we have 
  $\ds(D_1, D_2) \geq \sin\eps \geq \eps/2$.
  Therefore, for any~$\disks'\subseteq\disks$, any stabbing 
  of~$\disks'$ has length at
  least~$(|\disks'|-1)\cdot \frac{\eps}{2} = \Omega(|\disks'| / n^{1/d})$.  
\end{proof}

The arguments made in Section~\ref{sec:similar} imply the following
consequence of Theorem~\ref{thm:unbounded}.

\begin{theorem}
  There is a family~$\disks$ of~$n$ unit hyperdisks in~$\Reals^{d+1}$ 
  such that any convex container where we can pack $\disks$
  under translations has volume~$\Omega(n^{\frac{d-1}{d}})$.
  In particular, there is no convex container of bounded volume 
  into which all unit hyperdisks in~$\Reals^{d+1}$ can be packed under 
  translations.
\end{theorem}
\begin{proof}
  Consider the set~$\disks$ of hyperdisks of Theorem~\ref{thm:unbounded}.
  Let~$\OPT_C$ be the volume of an optimal convex container for~$\disks$.
  We have to show that $\OPT_C=\Omega(n^{\frac{d-1}{d}})$.
  
  Because of Lemma~\ref{lem:convexbox} there is a box~$M$ of
  volume at most~$(d+1)^{3(d+1)/2}\cdot \OPT_C$ where we can pack~$\disks$.
  We break~$\disks$ into $d+1$~groups $\disks_1,\dots,\disks_{d+1}$
  depending on which of the $d+1$ directions of the edges of~$M$ 
  minimize the angle with the normal.  
  Because of the pigeonhole principle, one of the groups, say $\disks_1$,
  has at least~$n/(d+1)$ hyperdisks. Let~$s$ be the corresponding
  direction of the edge of~$M$ minimizing the angle with the hyperdisks
  of~$\disks_1$. 
  Because of Theorem~\ref{thm:unbounded}, any stabbing of $\disks_1$
  has length at least $\Omega((n/(d+1))/n^{1/d}) = \Omega(n^{\frac{d-1}{d}})$.
  This is in particular true for the stabbing in the direction~$s$.
  Using Theorem~\ref{thm:similar} (in the second inequality of the 
  forthcoming chain), we get 
  \[
	(d+1)^{3(d+1)/2}\cdot \OPT_C \ge \vol(M) \ge \frac{\Omega(n^{\frac{d-1}{d}})}{2(d+1)^d} = 
    \Omega(n^{\frac{d-1}{d}}),
  \]
  which implies that $\OPT_C= \Omega(n^{\frac{d-1}{d}})$.
\end{proof}

We now prove the reverse direction for stabbing.
\begin{theorem}
  Let $\disks$ be a family of~$n$ unit hyperdisks in~$\Reals^{d+1}$ whose
  normals make an angle of at most~$\phi_0 = \arccos(1/\sqrt{d+1})$
  with the positive~$x_{d+1}$-axis.  Then there is a stabbing
  for~$\disks$ of length~$O(n^{\frac{d-1}{d}})$.
  \label{thm:stab-upper}
\end{theorem}
\begin{proof}
  Let~$S \subset \mathbb{S}^d$ be the set of directions of the unit normals
  of~$\disks$.  There is then a tour that visits the points of~$S$ one
  by one, and such that the sum of angles between consecutive points
  is bounded by~$O(n^{\frac{d-1}{d}})$.  This can be seen by first
  arguing that a minimum spanning tree of this cost
  exists~\cite{few1955shortest}, replacing that by an Euler tour, and
  finally shortcutting points that have been visited before.
  The claim now follows immediately from Lemma~\ref{lem:upsinxi}.
\end{proof}

Using Theorem~\ref{thm:general}, this immediately gives us the
following upper bound, which is now for the weaker class of
container, axis-parallel boxes.

\begin{corollary}
  Let $\disks$ be a family of~$n$ unit hyperdisks in~$\Reals^{d+1}$.
  Then there is an axis-parallel box of volume~$O(n^{\frac{d-1}{d}})$
  into which~$\disks$ can be packed.
\end{corollary}

\section{Other objects and open problems}
\label{sec:OtherObjects}

Our approximation algorithms can be extended to any arbitrary fixed
shape~$A$ of dimension~$d$, provided that~$A$ can be enclosed by some
hyperdisk~$D$ (that is, it is bounded) and contains another
hyperdisk~$d$ (that is, it has nonempty relative interior). More
precisely, if we are given a finite set of congruent copies of~$A$ in
$d+1$~dimensions, we can approximate the minimum-volume container
into which the set can be packed by translations.

This can be done by just applying our algorithm to the corresponding
set of copies of~$D$. Since this gives a constant-factor approximation
of the optimal packing of the~$D$'s, it also gives an approximation of
the optimal packing of the~$d$'s. Observe however, that the
approximation factor is multiplied by~$r^{d+1}$, where~$r$ is the ratio
between the radii of~$D$ and~$d$. Since the optimal packing of
the~$A$'s provides some packing of the~$d$'s its container must be at
least as large, from which we obtain an approximation for the~$A$'s.

The approximation factor obtained this way depends on the shape
of~$A$. For standard shapes such as squares ($r=\sqrt{2}$),
equilateral triangles ($r=2$) etc., we can directly compute it from
the approximation factors of the algorithms.

\medskip

We have given a constant-factor approximation for a special case of
objects, unit-radius hyperdisks in~$\Reals^{d+1}$.  It remains an open
problem whether an optimal packing of hyperdisks of different radii
can be efficiently approximated. It is also unclear what happens,
for example, when packing $(d-1)$-dimensional unit balls in~$\Reals^{d+1}$
under translations.
Finally, approximating the packing of arbitrarily oriented boxes or convex
polyhedra seems to be much more difficult.

\section*{Acknowledgments}

The authors thank G\"unter Rote for helpful comments and a
simplification in the proof of Theorem~\ref{thm:metric}.

N. Seiferth was partially supported by a fellowship within
the FITweltweit program and H.~Alt by the Johann-Gottfried-Herder
program, both of the German Academic Exchange Service
(DAAD). N.~Seiferth was also partially supported by the German
Science Foundation within the collaborative DACH project
\emph{Arrangements and Drawings} as DFG Project
MU3501/3-1. J.~Park was supported in part by starlab project
(IITP-2015-0-00199) and NRF-2017M3C4A7066317.
Research funded in part by the Slovenian Research and Innovation 
Agency (P1-0297, J1-2452, N1-0218, N1-0285).
Research funded in part by the European Union (ERC, KARST, project 
number 101071836). Views and opinions expressed are however those 
of the authors only and do not necessarily reflect those of the 
European Union or the European Research Council. Neither the 
European Union nor the granting authority can be held responsible for them.

\bibliography{references}

\end{document}